\documentclass[11 pt]{article}
\usepackage[utf8x]{inputenc}
\usepackage{geometry}
\usepackage{authblk} 
\usepackage{amsmath,amssymb,amsthm}
\usepackage{enumitem}
\usepackage{xspace}
\usepackage{hyperref}
\usepackage[capitalize]{cleveref}
\usepackage{caption}
\usepackage{subcaption}
\usepackage{xparse}
\usepackage{dsfont}
\usepackage[ruled,linesnumbered]{algorithm2e}

\usepackage{algpseudocode}
\usepackage{algorithmicx}
\usepackage{algcompatible}

\newtheorem{theorem}{Theorem}
\newtheorem{definition}[theorem]{Definition}

\newtheorem{lemma}[theorem]{Lemma}

\newtheorem{conjecture}[theorem]{Conjecture}
\newtheorem{claim}[theorem]{Claim}
\newtheorem*{theorem-non}{Theorem}

\usepackage{setspace}

\usepackage{todonotes}

\def\P{\mathcal{P}}

\def\C{\mathfrak{C}}
\def\O{\mathcal{O}}
\def\comps{\mathcal{S}}
\def\comp{S}

\def\coreSet{\mathcal{R}}

\def\centeri{\textit{center\xspace}}
\def\core{\textit{core }}

\def\tw{r }
\def\ra{t }

\def\shadow{\textit{shadow }}

\def\cov{covered\xspace}
\def\uncov{uncovered\xspace}
\def\vol{\texttt{Vol}\xspace}

\def\uncovset{\mathtt{Uncov}}

\def\og1{\overline{G^1}}

\def\core{R}

\def\centre{Y}
\def\shadow{Q}

\def\ball{\mathcal{B}}

\def\bag{B}

\def\barC{\bar{\C}}

\title{Approximate Max-Flow Min-Multicut Theorem for Graphs of Bounded Treewidth}

\begin{document}

\author{}
\author{Tobias~Friedrich}
\author{Davis~Issac}
\author{Nikhil~Kumar}
\author{Nadym~Mallek}
\author{Ziena~Zeif}
\affil{\normalsize
	\{firstname.lastname\}@hpi.de\\
	\vspace*{\baselineskip}
	Hasso Plattner Institute\authorcr
	University of Potsdam\authorcr
	Potsdam, Germany
}

\date{November 2022}

\maketitle

\begin{abstract}
We prove an approximate max-multiflow min-multicut theorem for bounded treewidth graphs. In particular, we show the following: Given a treewidth-$r$ graph, there exists a (fractional) multicommodity flow of value $f$, and a multicut of capacity $c$ such that $ f \leq c \leq \O(\ln (r+1)) \cdot f$. It is well known that the multiflow-multicut gap on an $r$-vertex (constant degree) expander graph can be $\Omega(\ln r)$, and hence our result is tight up to constant factors. Our proof is constructive, and we also obtain a polynomial time $\O(\ln (r+1))$-approximation algorithm for the minimum multicut problem on treewidth-$r$ graphs. Our algorithm proceeds by rounding the optimal fractional solution to the natural linear programming relaxation of the multicut problem. We introduce novel modifications to the well-known region growing algorithm to facilitate the rounding while guaranteeing at most a logarithmic factor loss in the treewidth.
\end{abstract}

\section{Introduction}

Given an undirected graph with edge capacities and $k$ source-sink pairs, the \emph{maximum multicommodity flow problem} asks for the maximum amount of (fractional) flow that can be routed between the source-sink pairs. Multicommodity flow problems (and its variants) have been studied extensively over the last five decades and find extensive applications in VLSI design, routing and wavelength assignment etc.~\cite{zang2000review}.

A natural dual to the maximum multicommodity flow problem is the \emph{minimum multicut problem}. Given an edge-capacitated graph with $k$ source-sink pairs, a multicut is a set of edges whose removal disconnects all the source-sink pairs, and the capacity (or value) of the cut is the sum of capacities of the edges in it. The value of any feasible multicommodity flow is at most the capacity of any feasible multicut. The ratio of the values of the minimum multicut and maximum multicommodity flow is called the \emph{multiflow-multicut gap}. Minimum multicut is NP-Hard to compute, even in very restricted settings such as trees~\cite{garg1997primal}. More precisely, it is known to be equivalent to the vertex cover problem in stars with unit weights~\cite{garg1997primal}, which implies that it is APX-Hard.

There is a rich literature on proving bounds on the multiflow-multicut gap. Perhaps the most famous of them is the max-flow min-cut theorem of Ford and Fulkerson~\cite{ford2009maximal}, which states that the value of the minimum multicut is equal to the maximum (integral) flow when $k=1$. Hu~\cite{hu1963multi} extended the result of Ford and Fulkerson to show that the multiflow-multicut gap is 1 even when $k=2$. There are many other special cases where the multiflow-multicut gap is 1, for example when $G$ is a path, but in general it can be arbitrarily large. Garg, Vazirani, and Yannakakis~\cite{garg1996approximate} proved a tight bound of $\Theta(\ln (k+1))$ on the multiflow-multicut gap for any graph $G$. If $G$ is a tree, then the  multiflow-multicut gap is exactly 2 \cite{garg1997primal}.

For $K_r$ minor-free graphs, Tardos and Vazirani~\cite{tardos1993improved} used the decomposition theorem of Klein Plotkin and Rao~\cite{klein1993excluded} to prove a bound of $\O(r^3)$ on the multiflow-multicut gap. This bound was subsequently improved to $\O(r^2)$ by Fakcharoenphol and Talwar \cite{fakcharoenphol2003improved}, and then to $\O(r)$ by \cite{abraham2014cops}. Given any natural number $n$, there exist graphs on $n$ vertices such that the multiflow-multicut gap on them is $\Omega(\ln (n+1))$ \cite{garg1996approximate}. This also implies a lower bound of $\Omega(\ln (r+1))$ for graphs which do not contain $K_r$ as a minor. It is conjectured by \cite{abraham2014cops} that this lower bound is tight (in fact they state their conjecture in terms of small diameter padded-decompostion and is slightly more general):
\begin{conjecture} \label{conjecture1}
Multiflow-multicut gap for $K_r$-minor free graphs is $\Theta(\ln( r+1))$.
\end{conjecture}
Since treewidth $r$ graphs do not contain $K_{r+2}$ as a minor, the above mentioned results also imply an upper bound of $\O(r)$ for treewidth $r$ graphs. In fact, \cite{abraham2014cops} show that their techniques can be extended to prove a bound of $\O(\log r + \log \log n)$ for treewidth-$r$ graphs. A lower bound of $\Omega(\ln (r+1))$ on the multiflow-multicut gap for graphs of treewidth $r$ follows form the discussion above. In this work, we show that Conjecture \ref{conjecture1} is true for graphs of bounded treewidth, which forms an important subclass of minor-free graphs. In particular, we show that the multiflow-multicut gap for graphs of treewidth $r$ is $\Theta(\ln (r+1))$. Our proof is constructive, and we also obtain a polynomial time $\O(\ln (r+1))$ approximation algorithm for the minimum multicut problem on graphs of treewidth-$r$. 
\begin{theorem} \label{theorem: flow-cut gap}
    Let $G=(V,E)$ be a treewidth-$r$ graph with edge capacity $c:E \rightarrow \mathcal{R}_{\geq 0}$ and $(s_i,t_i), 1 \leq i \leq k$ be a set of source-sink pairs. Then there exists a polynomial time computable multicommodity flow of value $f$ and a multicut of value $c$ such that $f \leq c \leq \O(\ln (r+1)) \cdot f$.
\end{theorem}


\section{Preliminaries} \label{section: preliminaries}
Let $G=(V,E)$ be a simple undirected graph with edge capacities $c\colon E\rightarrow\mathds{Z}_{\geq 0}$. Let $(s_i,t_i)$ be the source-sink pairs.
Let $\P$ be the set of all paths in $G$ between a source and its corresponding sink. A multiflow $f\colon \P\rightarrow\mathds{R}_{\geq 0}$ is feasible if for every edge $e\in E$, the total flow on all paths containing the edge, $\sum_{P:e\in\P} f_P$, is at most the capacity of the edge, $c(e)$. For a path $P\in\P$, we refer to $f_P$ as the value of flow on $P$.
A maximum multiflow is a feasible flow $f$ which maximises $\sum_{P\in\P} f_P$. A multicut is a set of edges $E'\subseteq E$ such that every $P\in\P$ contains at least one edge in $E'$. Equivalently, a multicut is a set of edges whose removal disconnects every source-sink pair. Since a multicut contains an edge of every path in $\P$, the value of any feasible multicut is at least the value of any feasible multiflow. The ratio of the minimum multicut to the maximum multiflow is called the multiflow-multicut gap.

A cut $S\subseteq V$ is a partition of the vertex set $(S,V\setminus S)$. Let $\delta_{E}(S)$ denote the edges in $E$ with exactly one endpoint in $S$. We will usually drop the subscript $E$ if it is clear from the context. We may also subscript by the graph and say $\delta_G(S)$. For a subset $E'\subseteq E$ let $c(E')$ be the total capacity of edges in $E'$. Given a graph $G=(V,E)$ with edge-length $l:E \rightarrow \mathds{R}_{\geq 0}$, we denote the length of the shortest $u,v$ path in $G$ (w.r.t $l$) by $d_{G}(u,v)$,  where we may omit the subscript $G$ if it is clear from the context.
Similarly for a set $S$ we use $d_{G}(u,S)$ or $d(u,S)$ to denote the shortest path between $u$ and $S$.
Whenever we say distance between two vertices we mean the distance $d$ ie.~the shortest path distance w.r.t to $l$. For graph $G=(V,E)$ and $V'\subseteq V$ we will use $G-V'$ to denote the graph induced in $G$ by $V\setminus V'$.
For a vertex $s$ and value $\alpha$, we use $\ball_G(s,\alpha)$ to denote the set of vertices that are a distance of at most $\alpha$ from $s$.
Also, for a vertex set $S$ we use $\ball_G(S,\alpha)$ to denote the set of vertices that are a distance of at most $\alpha$ from $S$.
We call $\ball_G(S,\alpha)$ as the ball of radius $\alpha$ centered around $S$ in $G$. We will omit the subscript $G$ when it is clear.

A graph $G=(V,E)$ is said to have treewidth at most $r$ if there exist subsets $S_1,S_2,\ldots,S_m$ of the vertex set $V$ called bags and a tree $T$ with $S_1,S_2,\ldots,S_m$ as vertices such that (i) $|S_i| \leq r+1$ for $i \in [1,m]$ (ii) for each $(u,v) \in E$, there exists a $j \in [1,m]$ such that $u,v \in S_j$ (iii) for each $v \in V$, the subgraph induced by $T_v=\{S_j \mid v \in S_j\}$ on $T$ is connected. It follows from the definition that trees have treewidth-1 and any graph on $n$ vertices has treewidth at most $n$. \cite{feige2005improved} gave a polynomial time algorithm that computes a bag decomposition of width $\mathcal{O}(r \cdot \sqrt{\ln(r+1)})$, if $G$ has a treewidth of at most $r$. We use the above algorithm to compute a bag-decomposition of size $\mathcal{O}(r \cdot \sqrt{\ln(r+1)})$. Since all our guarantees are a logarithmic function of the bag-decomposition size, this approximation does not have any asymptotic effect on our guarantees. For the sake of presentation, from now on we assume that the tree decomposition of the graph is given to us.

\section{Linear Programming Formulation for Mulicut}
We first describe an integer programming formulation for the minimum multicut problem. For every edge $e \in E$, we have a variable $x(e)$, which indicates if the edge is picked in the cut. We want to disconnect every path between the source-sink pairs, therefore we have a constraint saying that at least one edge must be picked on every path between a source-sink pair. Let $d(u,v)$ denote the distance between the vertices $u$ and $v$ by setting the length of each edge $e \in E$ to be equal to $x(e)$. We drop the integrality constraints to obtain a linear programming (LP) relaxation for the multicut problem. The dual of the multicut LP is exactly the maximum multicommodity flow problem \cite{garg1996approximate}. In particular, we have a non-negative flow variable $f_P$ for each $P \in \mathcal{P}$, and capacity constraints for each edge. 

\vspace{0.2cm}

\begin{minipage}{0.49\textwidth}
    \begin{align*}
        \min \displaystyle\sum_{e \in E}& c(e)x(e)\\[3pt]
        d(s_i,t_i) \geq 1 & \text{ for } 1 \leq i \leq k\\[11pt]
        0 \leq  x(e) & \leq 1 \text{ for } e \in E
    \end{align*}
\end{minipage}
\begin{minipage}{0.49\textwidth}
    \begin{align*}
        \max \displaystyle \sum_{P \in \mathcal{P}} & f_P\\
        \displaystyle\sum_{P:e \in P} f_P  \leq  c(e) & \text{ for } e \in E\\    
        f_P \geq  0 \text{ for } & P \in \mathcal{P}
    \end{align*}
\end{minipage}

\vspace{0.4cm}
Even though there are an exponential number of constraints, it is well known that the optimal solution to the above LPs can be computed in polynomial time \cite{garg1996approximate}. We refer to the optimum solution to the above LP as the minimum fractional multicut and maximum (fractional) multicommodity-flow. By the strong duality theorem, the minimum fractional multicut is exactly equal to the maximum multicommodity flow that can be routed between the source-sink pairs. If we can construct a multicut with value at most $\alpha \geq 1$ times the minimum fractional solution, we obtain an $\alpha$-approximation algorithm for the minimum multicut problem. Furthermore, since the value of any feasible multicut is at least the value of the minimum fractional multicut, we also obtain the following approximate max-flow min-cut theorem:
\begin{center}
    $\text{max-multicommodity-flow} \leq \text{min-multicut} \leq \alpha \cdot \text{max-multicommodity-flow}$
\end{center}
Garg, Vazirani and Yannakakis~\cite{garg1996approximate} gave a region growing algorithm to obtain a multicut of value at most $\O(\ln (k+1))$ times the minimum fractional solution. The region growing algorithm was first introduced by Leighton and Rao~\cite{leighton1999multicommodity} in the context of sparsest-cut problem. Our algorithm also builds upon this idea. We next describe the algorithm by \cite{garg1996approximate} in more detail. 

\section{Region Growing Algorithm}
\label{sec:region}
As before, suppose we are given a graph $G=(V,E)$ with edge capacities $c:E \rightarrow \mathds{R}_{\geq 0}$. Furthermore, we also have a length function on the edges $l: E \rightarrow \mathds{R}_{\geq 0}$. 
Recall that for $S \subseteq V$, $\ball(S,\ra)$ denotes the set of vertices whose distance from $S$ is at most $\ra$. We call $\ball(S,\ra)$ the \emph{ball of radius $\ra$ centred around $S$}. 
Let $\vol_G(S,\ra)$ be the \emph{volume} of the edges contained inside the ball $\ball(S,\ra)$, defined as follows

\begin{equation}
    \label{equation::regionGrowing}
    \vol_G(S,\ra)=\vol(S,0)+\displaystyle\sum_{\substack{(u,v) \in E, \\ u,v \in \ball(S,\ra)}}c(u,v) \cdot l(u,v)+
    \displaystyle\sum_{\substack{(u,v) \in E, \\
    u \in \ball(S,\ra), v \notin \ball(S,\ra)}}c(u,v) 
    \cdot (\ra-d(S,u)) 
\end{equation}
Here, $\vol(S,0)$ is the \emph{initial volume} at $S$ and can be chosen to be any positive real number. 
We also define $\vol'$ to be the volume without the initial volume ie.,
\begin{align*}
    \vol'_G(S,t) = \vol_G(S,t)-\vol(S,0).
\end{align*}
Let $C_G(S,\ra)$ be the total capacity of the cut edges going across the set $\ball(s,\ra)$. More formally, 
\begin{equation*}
 C_G(S,\ra)= \displaystyle\sum_{(u,v) \in \delta(\ball(s,\ra))}c(u,v)   
\end{equation*}
In the above definitions, the subscript $G$ may be omitted when the graph is clear from context.
Given $a, b \in \mathds{R}_{\geq 0}$, we will use $r\sim \mathds{U}(a,b)$ to denote the fact that $\ra$ is chosen uniformly at random in the interval $(a,b)$. \cite{garg1996approximate} showed that if $\ra \sim \mathds{U}(a,b)$, then the expected value of $C(S,\ra)$ can be bounded in terms of $\vol(S,\ra)$. 
\begin{lemma} $\displaystyle \mathop{\mathds{E}}_{\ra \sim \mathds{U}(a,b)} \left( \dfrac{C(S,\ra)}{\vol(S,\ra)} \right) \leq  \dfrac{1}{b-a} \cdot \ln  \left( \dfrac{\vol(S,b)}{\vol(S,a)} \right)$.

\end{lemma}
Since there are at most $n$ distinct balls $\ball(S,\ra)$ for any $S \subseteq V$, the above also implies the existence of a polynomial time computable $\ra_0 \in [a,b)$ such that 

\begin{equation}
    \label{equation::estimateCost}
    C(S,\ra_0)  \leq  \dfrac{1}{b-a} \cdot \ln  \left( \dfrac{\vol(S,b)}{\vol(S,a)} \right) \cdot \vol(S,\ra_0)
\end{equation}
We will refer to this as the \textbf{region growing lemma} \footnote{we follow the presentation in Section 8.3 of the book by Williamson-Shmoys \cite{williamson2011design}}. \cite{garg1996approximate} use the region growing lemma to construct a multicut as follows. First they find an optimal solution to the linear programming relaxation for the multicut, say $\{x_{e}^{*}\}_{e \in E}$, and set $l(e)=x^{*}(e)$. The algorithm picks edges into the multicut as follows: if there exists a component containing an $(s_i,t_i)$ pair, then choose $\ra \in [0,1/2)$ as guaranteed by the region growing lemma and include the edges in $\delta(\ball(s_i,\ra))$ into the multicut. Since the diameter of the graph induced by vertices in $\ball(s_i,\ra)$ is at most $1$, it does not contain any source-sink pair. Hence, we can safely remove the vertices in $\delta(\ball(s_i,\ra))$ from the graph and iterate on the rest of the graph. The algorithm terminates when no connected component contains an $(s_i,t_i)$ pair. Note that the algorithm picks at most $k$ cuts using region growing. They set the initial volumes $\vol(s_i,0)=V^{*}/k$ for $i=1,2,\ldots,k$, where $V^{*}=\sum_{e \in E}c(e) \cdot x^{*}(e)$. Using the region growing lemma, the cost of the solution can then be bounded by $\mathcal{O}(\ln (k+1)) \cdot V^{*}$. 

\section{Multicuts and Small Diameter Decomposition} \label{Section:Phase_3}
Let $G=(V,E)$ be a graph with edge capacities $c:E \rightarrow \mathds{R}_{\geq 0}$ and edge lengths $l: E \rightarrow \mathds{R}_{\geq 0}$. A subset of edges $E' \subseteq E$ is called a \textbf{small diameter decomposition} of $G$ if all the components of $G'=(V,E \setminus E')$ have diameter strictly less than 1 with respect to the distance $d_G$ (notice that the distance is not in $G'$)\footnote{this is also known as the weak-diameter decomposition}. The cost of the decomposition is the sum of all the edge capacities in $E'$, ie.~$c(E')$. Note that the small diameter decomposition corresponds to a feasible multicut when the length of edges correspond to a feasible (fractional) solution to the linear programming relaxation of the multicut. If $G$ has treewidth at most $r$, then we will give a polynomial time algorithm to find a small diameter decomposition whose cost is $\mathcal{O}\left(\ln (r+1)\right) \cdot \sum_{e \in E}c(e) \cdot l(e)$. This will give us an $\mathcal{O}\left(\ln (r+1)\right)$-approximation algorithm for the minimum multicut problem, and also imply a bound of $\mathcal{O}\left(\ln (r+1)\right)$ on the multiflow-multicut gap for graphs with treewidth at most $r$. We first show that if there exists a set $S \subseteq V$ of at most $r$ vertices such that for every vertex $v \in V$, there exists a $u \in S$ such that $d(u,v) \leq 1/4$, then there exists a small diameter decomposition with cost $\mathcal{O}(\ln (r+1)) \cdot \sum_{e \in E}c(e) \cdot l(e)$. This lemma will serve as an important building block in our algorithm for graphs of bounded treewidth.
\begin{lemma} \label{lemma:sdd_width_0}
Let $G=(V,E)$ be a graph with edge capacities $c:E \rightarrow\mathds{R}_{\geq 0}$ and edge lengths $l: E \rightarrow \mathds{R}_{\geq 0}$. Let $S \subseteq V$ be such that $|S| \leq r$ and for all $u \in V$, $d(u,S) \leq 1/4$. Then there exists a small diameter decomposition of $G$ with cost at most $1/8 \cdot \ln (r+1) \cdot F$, where $F=\sum_{e \in E}c(e) \cdot l(e)$.
\end{lemma}

\begin{proof}
We may assume w.l.o.g.~that $S$ has size exactly $r$. Let $S=\{s_1,s_2,\ldots,s_r\}$. We will use the region growing lemma to give such a decomposition. We will have $r$ iterations, and after each iteration, we will construct a new graph $G_i= (V_i, E_i)$ and a set of vertices $S_i$. The vertices in $S_i$ will have the property that $d_G(u,S_i) \leq 1/4$ for each vertex $u \in V_i$ and $|S_i|=r-i$. Initially, we set $G_1=G$ and $S_1=S$. In iteration $i$, we use the region growing lemma in the graph $G_i$ to pick a ball centered at $s_i \in S_i$ with radius $ \ra_i \in [1/4,1/2)$ such that:
\begin{equation*}
C(s_i,\ra_i)  \leq  4 \cdot \ln  \left( \dfrac{\vol(s_i,1/2)}{\vol(s_i,1/4)} \right) \cdot \vol(s_i,\ra_i)    
\end{equation*}
We set $\vol(s_i,0)=F/r$ and let $F_i=\vol(s_i,\ra_i)-\vol(s_i,0)$. Since $\vol(s_i,1/2) \leq F + \vol(s_i,0) = F +F/r$ and $\vol(s_i,1/4) \geq \vol(s_i,0) \geq F/r$, we have 
\begin{equation*}
C(s_i,\ra_i)  \leq  4 \cdot \ln (r+1) \cdot \vol(s_i,\ra_i) \leq  4 \cdot \ln (r+1) \cdot (F_i + F/r).   
\end{equation*}
We include all the edges in $\delta(\ball(s_i,\ra_i))\cap E$ in our small diameter decomposition. We then go on to construct a new graph $G_{i+1}=(V_{i+1},E_{i+1})$ by first removing the vertices in $\ball(s_i,\ra_i)$ from $G_i$. If some $s_j \in S_i \setminus \{s_i\}$ is contained in the set $\ball(s_i,\ra_i)$, then we {(re-)introduce} the vertex $s_j$ to $G_{i+1}$ and connect it to every other vertex $u \in V_i \setminus \ball(s_i,\ra_i)$ by an edge of capacity $0$ and length $d_G(s_j,u)$. It is easy to observe that in $G_{i+1}$, we have a set $S_{i+1}=\{s_{i+1},\ldots,s_r\}$ of $r-i$ vertices such that $d_G(u,S_{i+1}) \leq 1/4$ for all $u \in V_{i+1}$.

We run the above procedure on all the vertices in $S$. Since every vertex of $G$ is at a distance of at most 1/4 from at least one of the vertices in $S$, and we pick a ball of radius at least 1/4 from each of the vertices in $S$, each vertex of $G$ is a part of at least one of the $\ball(s_j,\ra_j)$. Also, each vertex in $\ball(s_j,\ra_j)$ is at a distance of strictly less than $1/2$ from $s_j$ and hence any pair of vertices in $\ball(s_j,\ra_j)$ has a distance less than one. Hence, each vertex is contained in a connected component with diameter at most one (w.r.t.~$d_G$), and we have a feasible small diameter decomposition. Using the fact that $\sum_{i=1}^{r}F_i \leq F$, the total cost of the small diameter decomposition can be bounded by: 
\begin{equation*}
   \displaystyle\sum_{i=1}^{r} C(s_i,\ra_i) \leq 4 \cdot \ln (r+1) \cdot \displaystyle\sum_{i=1}^{r} \left(F_i + \frac{F}{r}\right) = 4 \cdot \ln (r+1) \cdot \left(F + \displaystyle\sum_{i=1}^{r}F_i\right) \leq 8 \cdot \ln (r+1) \cdot F 
\end{equation*}\qedhere
\end{proof}

\section{Graphs of Width $r$}
To simplify the presentation, we will work with a class of graphs that is slightly more general than graphs of treewidth $r$.
\begin{definition}
\textsc{Graphs of width $r$}: A graph $G=(V,E)$ is said to have width $r$ if there exists a partition $\mathcal{S}=\{S_1,S_2,\ldots,S_m\}$ of $V$ and a rooted tree $T$ with $S_1,S_2,\ldots,S_m$ as its vertices such that $|S_i|\leq r$ for each $S_i \in \mathcal{S}$ and for each $(u,v) \in E$, there exists a $S_i,S_j \in \mathcal{S}$ such that $S_j$ is a parent of $S_i$ and $u,v \in S_i \cup S_j$.
The tree $T$ is called the width-$r$ decomposition of $G$.
\end{definition}
Notice that unlike in a tree decomposition, here the bags $S_i$ are disjoint from each other. As before, let $G=(V,E)$ be a graph with edge capacity $c:E \rightarrow \mathds{R}_{\geq 0}$, edge-length $l: E \rightarrow \mathds{R}_{\geq 0}$ and $F=\sum_{e \in E}c(e) \cdot l(e)$. It is easy to see that if $G$ has treewidth $r$, then we can construct an equivalent graph $G'=(V',E')$ with width $r+1$ as follows: for each $u \in V$, if $u$ appears in bags $B_1,B_2,\ldots,B_k$ in the tree decomposition of $G$, then we replace each appearance of $u$ in the bags by a new (distinct) vertex, ie.~$u_1,u_2,\ldots,u_k$ respectively. Furthermore, we connect $u_i \in B_i,u_j \in B_j$ by an edge of sufficiently high capacity and zero length if the bags $B_i, B_j$ are adjacent in the tree decomposition of $G$. For each $(u,v) \in E$, there exists at least one bag $B_j$ of the tree decomposition of $G$ such that $u,v \in B_j$. We add the edge $(u_j,v_j)$ to $E'$. If there are multiple bags with the edge, we add only one of the edges. It is straight forward to verify that $G'$ has width at most $r+1$ and $F=\sum_{e \in E}c(e) \cdot l(e)=\sum_{e \in E'}c(e) \cdot l(e)$.
An illustration of the transition from a tree decomposition to a width decomposition can be found in \cref{fig:treeToWidth}.
\vspace{0.3cm}

\begin{figure}[ht]
    \begin{center}
       \includegraphics[scale=0.3]{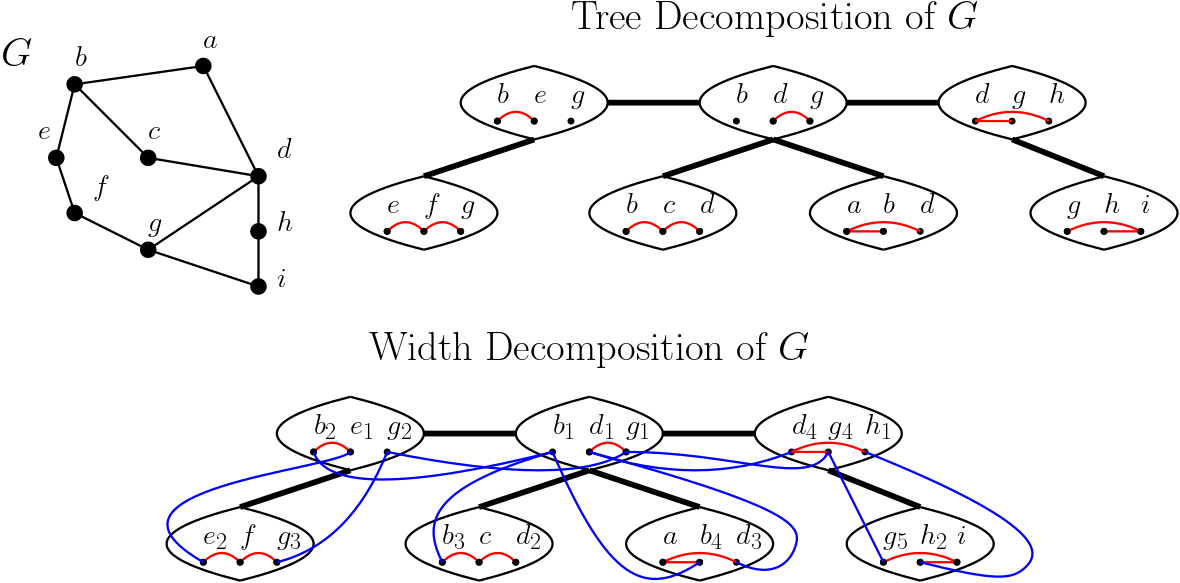} 
    \end{center}
    \caption{Transition from a tree decomposition to a width decomposition. The red edges are the original edges in $G$ and the blue edges are the added edges with length zero and high capacity.}
    \label{fig:treeToWidth}
\end{figure}

Any small diameter decomposition of $G'$ does not split copies of the same vertex in $G$ (as they are connected with edges of very high capacity), and hence, also corresponds to a small diameter decomposition for $G$ with the same cost. In the next section, we give an algorithm for computing a small diameter decomposition for a graph of width $r$.


\section{Overview of The Algorithm}

Our algorithm has three phases. In phase one, we build a collection of sets, called the $\textit{cores}$. The cores form a cover of the vertices in $G$. 
Furthermore, for every core $\core$, we have a set $\centre$ of at most $\tw$ vertices such that each vertex of $\core$ is close to at least one of the vertices in $\centre$. We refer to $\centre$ as the \textit{center} of the core. In phase two, we process the cores in a top down order w.r.t the width-$r$ decomposition.  When processing a core, we use the region growing algorithm of \cite{garg1996approximate} to pick a set of cut edges and remove these edges from the graph.
Each connected component obtained at the end of phase two has an associated core and center. Each vertex in the connected component is close in the original graph $G$ to one of the at most $\tw$ vertices of its center (though center might not be contained in the component). Finally, in phase three, we use the algorithm described in Section \ref{Section:Phase_3} for each connected component to find a small diameter decomposition. We prove the following, from which Theorem \ref{theorem: flow-cut gap} follows as well. 
\begin{theorem} \label{Theorem: main}
    Let $G=(V,E)$ be a graph of width at most $\tw$ with edge capacities $c:E \rightarrow \mathds{R}_{\geq 0}$ and edge lengths $l: E \rightarrow \mathds{R}_{\geq 0}$. Then there exists a small diameter decomposition of $G$ with cost $\O(\ln (\tw+1)) \cdot F$, where $F=\sum_{e \in E}c(e) \cdot l(e)$. Moreover, such a decomposition can be constructed in polynomial time. 
\end{theorem}



\section{The Algorithm}
\begin{figure}[ht] \label{definitions}
\begin{center}
   \includegraphics[scale=0.5]{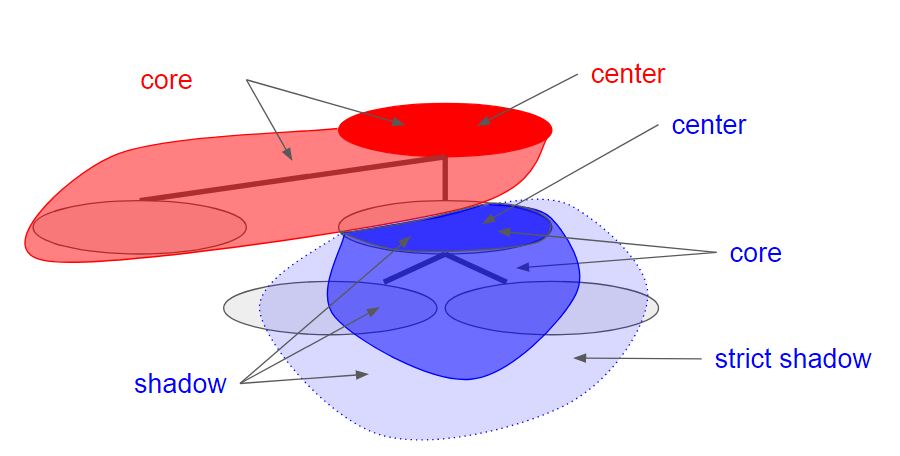} 
\end{center}
\caption{An illustration of the terms used in the algorithm and its analysis.}
\end{figure}
Let $G=(V,E)$ be a graph with edge-capacity $c:E \rightarrow \mathds{R}_{\geq 0}$ and edge-length $l: E \rightarrow \mathds{R}_{\geq 0}$. 
We also have two parameters $a,b \in \mathbb{R}_{\geq 0}$ with $b > a$ for the algorithm.
They will in fact be fixed to $a=1/8$ and $b=1/4$ later.

Let $T(G)$ be the width-$r$ decomposition of $G$.
We may omit the brackets and use just $T$.
We define the level of a bag in $T(G)$ to be its hop-distance from the root of $T(G)$.
For a bag $B$ of a subtree $T'$ of $T(G)$, we define $T'_B$ to be the subtree of $T'$ rooted at $B$.
For a subtree $T'$ of $T(G)$, we define
$V[T']$ to be the union of all bags in $T'$.


\medskip
\textbf{Phase 1. Growing Cores}: The first phase of the algorithm outputs a set $\coreSet$ of subsets of the vertex set $V$ whose union covers $V$. Each set in $\coreSet$ is called a \emph{core}.

During the algorithm, we say that a vertex is \emph{covered} if it is part of at least one core constructed so far.
For a $V'\subseteq V$, we use $\uncovset(V')$ to denote the set of uncovered vertices of $V'$.
Also, we say that a bag of $T(G)$ is \emph{covered} if all the vertices in the bag are covered. Similarly, a bag is \uncov if one of its vertices is uncovered.

We associate with each bag $B$ of $T(G)$ an attachment $A(B)\subseteq V$.
For a subtree $T'$ of $T(G)$ we use $A[T']$ to denote the union of attachments of all bags of $T'$.

\emph{Initialization:} We initialize the set of cores $\coreSet$ to $\emptyset$. 
The attachment of each bag of $T(G)$ is initialized to $\emptyset$.

We proceed in 
\emph{iterations} during Phase 1 until all vertices are covered. During each iteration we process one by one each connected component $T'$ of the forest induced on $T(G)$ by the uncovered bags of $T(G)$ (note that $T'$ gives a rooted subtree of $T(G)$).
We process each such $T'$ in a top-down manner as follows (see \cref{fig:phase1}):
\begin{enumerate}
    \item Mark all bags of $T'$ as \emph{unvisited}.
    
    \item While there is an unvisited bag in $T'$:
    \begin{enumerate}
        \item Pick an unvisited bag $B$ of $T'$ of the smallest level (breaking ties arbitrarily).
        \item In the graph induced in $G$ by $
        \uncovset(V[T'_B])
        \cup A[T'_B]$ 
       we pick the ball $\ball(\uncovset(B),a)$ as a new \emph{core} $\core$ into $\coreSet$.
       \item Mark all the bags of $T'$ that intersects $\core$ as \emph{visited}.
       \item If $B$ is not the root of $T'$ then add $\core$ to the attachment of the parent bag of $B$.
    \end{enumerate}
\end{enumerate}

This finishes the Phase 1 algorithm. The 
\emph{center} of a core is defined as the center from which the ball defining the core was picked during Phase 1.
The following lemmas follow directly from the construction.

\begin{lemma}
Every vertex $v$ of $G$ is in at least one core in $\coreSet$.
\end{lemma}

\begin{lemma}
The center of each core is contained in some bag of $T(G)$. Also, each bag contains the center of at most one core.
\end{lemma}

\begin{lemma}
    \label{lemma::coresDisjointRank}
    The cores constructed in the same iteration are vertex disjoint.  
\end{lemma}

We call the bag containing the center of a core to be its \emph{center-bag}. We now bound the number of iterations of Phase 1.
%
\begin{figure}[h]\label{shadows}
\centering
\begin{subfigure}{.5\textwidth}
  \centering
  \includegraphics[width=1\linewidth]{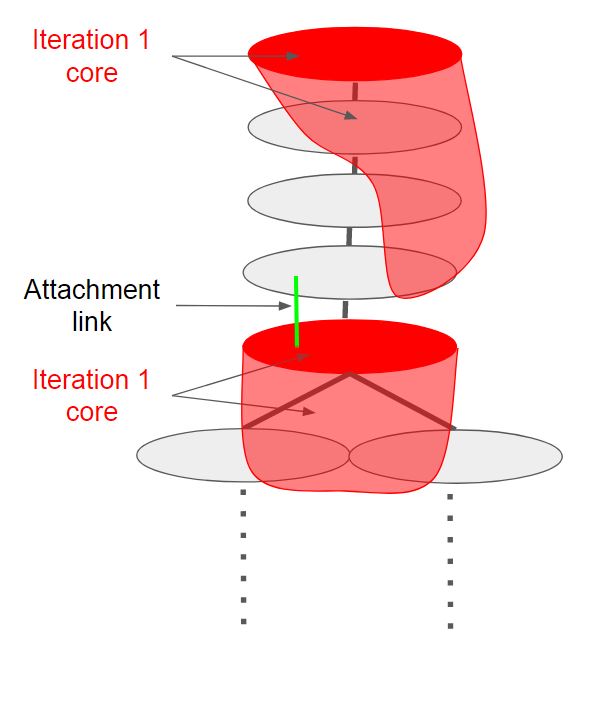}
  \caption{}
  \label{fig:sub1}
\end{subfigure}%
\begin{subfigure}{.5\textwidth}
  \centering
  \includegraphics[width=1\linewidth]{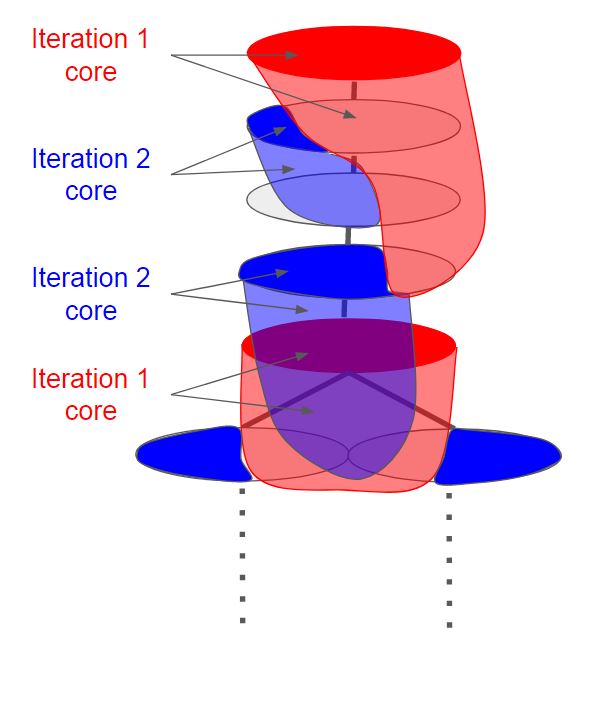}
  \caption{}
  \label{fig:sub2}
\end{subfigure}
\caption{An illustration of two iterations of Phase 1. Note that the attachment links allow a core of iteration 2 to grow within a core of iteration 1.}
\label{fig:phase1}
\end{figure}


\begin{lemma}
The Phase 1 of the algorithm has at most \tw iterations.
\label{lem:numberiterations}
\end{lemma}

\begin{proof}
At each iteration of the Phase 1, at least one vertex per bag of each uncovered bag of $T(G)$ gets $\cov$. Since each bag has no more than $\tw$ elements, the algorithm can not have more than $\tw$ iterations.
\end{proof}
\medskip
\textbf{Phase 2. Growing Components}: 
The goal of the second phase is to partition the graph $G$ into components
$\comps$ such that each component $\comp\in \comps$ has a set $\centre$ of $r$ vertices in the original graph (not necessarily in $\comp$) called its center such that each vertex in $\comp$ is at a distance of at most $b$ in $G$ from $\centre$.
The components will be such that the capacity of edges going across the components is small.

We will process the cores in $\coreSet$ in a specific top-down order. While processing a core, we will grow a ball with the core as center, and pick this as one component. 
The center of each component will be in fact the center of the core from which it was grown.
The capacity of the edges going across the components can be bounded as we will pick the radius of the ball as given by the region growing lemma.

During the second phase algorithm, we will process the cores in a top-down manner. That is, out of all the unprocessed cores we pick the one whose center-bag is in the smallest level in $T(G)$ (breaking ties arbitrarily) to be processed next. Let $\coreSet=\{\core_1,\ldots,\core_p\}$ be the cores in this top-down order.
Each core $\core_i$ is processed as follows:

\smallskip
    We pick a new component $\comp_i$ grown from $\core_i$ as follows: 
    let $G_i:=G_{i-1}-\comp_{i-1}$ where $G_1=G$.
    Also, let $R'_i$ be defined as $R_i-\comp_1-\ldots -\comp_{i-1}$.
    If $R'_i$ is isolated in $G_i$ then take $\comp_i:=R'_i$. Otherwise take $\comp_i$ as the ball $\ball_{G_i}(R'_i,\ra_i)$ where $\ra_i\in [0,b-a)$ is given by the region growing lemma (see \cref{equation::estimateCost}) such that   
    \begin{align}
    C_{G_i}(\core'_i,\ra_i)  \leq  \dfrac{1}{b-a} \cdot \ln  \left( \dfrac{\vol_{G_i}(\core'_i,b-a)}{\vol(\core'_i,0)} \right) \cdot \vol_{G_i}(\core'_i,\ra_i),
    \label{eq:phase2}
    \end{align}
    where the initial volume $\vol(\core'_i, 0)$ is set to be $\vol'_{G_i}(\core'_i,b-a)/h$ where $h=2r^3+2r$.
\smallskip


%
%

That concludes the processing of a core and also concludes the Phase 2 algorithm.
Let $\comps=\{\comp_1,\ldots,\comp_p\}$ denote the set of components. 
Let $X_2$ be the set of edges of $G$ going across the components in $\comps$.
We say that $X_2$ is the set of cut-edges picked in Phase 2.
In \cref{section::boundWeight}, we will bound their total capacity to be logarithmic in the width of the graph.
An important property that is used for this is that if $b=2a$ each edge contributes to the volume $\vol'_{G_i}(R'_i,b-a)$ of at most $\O\left(r^3\right)$ many cores. 
This property will be proved in \cref{sec:shadow}.
We now state the required \emph{near to center} property of components, which follows just by construction.
\begin{lemma}
\label{lem:nearcentre}
Let $\comp$ be the component constructed during the processing of core $\core$.
Let $\centre$ be the center of $\core$. All vertices in $\comp$ are at a distance of at most $b$ from $Y$ in $G$.
\end{lemma}
Also, we prove that the components indeed give a partition of the vertex set.
\begin{lemma}
Each vertex $v$ of $G$ is in exactly one component in $\comps$.
\end{lemma}
\begin{proof}
Since the vertices in $G_i$ does not intersect any component of $S_1,\ldots,S_{i-1}$, we have that $v$ cannot be in more than one component.
To see that $v$ should be in a component, recall that it should be in at least one core $R$.
When this core $R$ is processed, $v$ will be picked into the resulting component, if it has not been picked in any component so far.
\end{proof}

\medskip
\textbf{Phase 3. Decomposing Components}: 
Each component $\comp\in \comps$ output by the second phase has the property that there is a set $\centre\subseteq V$ called the center of $\comp$ containing at most $r$ vertices (not necessarily in $\comp$) such that each vertex in $\comp$ is at most a distance of $b$ away from $\centre$.
In the third phase for each component $\comp$ we make an auxillary graph $G(S)$ on which we apply 
\cref{lemma:sdd_width_0}. The graph $G(S)$ is obtained by taking the induced graph of the component $G[S]$ and adding the vertices $\centre$ to it. 
In addition, for each $y\in \centre$ and each $s\in \comp$, we add edges of capacity $0$ and length equal to $d_G(y,s)$ to $G(S)$.
In this auxillary graph we find a small diameter decomposition of small cost by using \cref{lemma:sdd_width_0}.
We satisfy the pre-condition of \cref{lemma:sdd_width_0}, if we set $b=1/4$, by using \cref{lem:nearcentre}. 
Note that this means we should set $a=1/8$ as the condition $b=2a$ is required for bounding the cost of cut-edges.
Let $X'(S)$ be the cut-edges of the small diameter decomposition of $G(S)$ given by \cref{lemma:sdd_width_0}, and let $X(S)$ be $X'(S)$ minus the auxilliary edges from $Y$ to $S$.
Note that the cost of $X(S)$ is same as the cost of $X'(S)$ and $X(S)$ gives a small diameter decomposition of $G[S]$ with respect to the distance $d_G$.
Our final set of cut-edges giving a small diameter decomposition of $G$ is given by the union of $X_2$ (the cut-edges from Phase 2) and $X_3:=\bigcup_{\comp\in \comps}X(\comp)$.

\begin{lemma}
\label{lem:smalldiameterdecomp}
The set of cut-edges $X_2\cup X_3$ gives a small diameter decomposition of $G$.
\end{lemma}
\begin{proof}
Once we remove $X_2$ there are no edges going across the components in $\comps$ output by Phase 2.
Within each of those components \cref{lemma:sdd_width_0} guarantees that after removing $X(S)$ the remaining connected components have a diameter strictly less than one with respect to the distance $d_G$.
Thus we have a small diameter decomposition of $G$.
\end{proof}

We will bound the cost of $X_2\cup X_3$ in the next two sections.

\section{Bounding the Volume Contributions of an Edge}\label{sec:shadow}
We want to upper bound the capacity of edges going across components produced by Phase 2.
For this, it is clearly sufficient to bound the capacity of $\delta(S_i)$ in the graph $G_i$, for each core $R_i$, because
when we process a core $R_i$ in Phase 2, we consider only the graph $G_i$.
Thus we want to bound the capacity of $\delta_{G_i}(\ball_{G_i}(\core'_i,\ra_i))$, where $\ra_i < b-a$ is the radius that satisfies \cref{eq:phase2}.
In \cref{eq:phase2}, notice that to bound the capacity, we use the volume of the ball $\ball_{G_i}(\core'_i,b-a)$, whereas we remove only $\comp_i=\ball_{G_i}(\core'_i,\ra_i)$ afterwards.
Thus, it is possible that the edges of the graph $G[\ball_{G_i}(\core'_i, b-a) \setminus \ball_{G_i}(\core'_i,\ra_i)]$ are used to pay for certain cut edges for components arising in the processing of subsequent cores. 
Our goal in this section is to bound the number of times an edge can be used to pay for the cut-edges of a component with respect to the width
of the graph, as described in the following lemma.

\begin{lemma}
\label{lem:edgeVolbound}
Let $e$ be an edge of $G$.
The number of cores $\core_i$ for which $e$ appears (even partially) in $\ball_{G_i}(\core'_i,b-a)$ during the processing of $\core_i$ in Phase 2 is at most $2r^3+2r$.
\end{lemma}

The rest of this section is devoted to proving the above lemma.
We begin by introducing some definitions.
We say that a core has \emph{rank} $i$ if it was constructed in the $i$-th iteration of phase 1. 
We define the rank of a vertex to be the lowest rank among all the cores containing it.
For each core we define its \emph{center-bag} to be the bag of $T(G)$ that contains its \centeri.
We say that a core $\core_1$ is an ancestor (descendant resp.) of $\core_2$ if the center-bag of $\core_1$ is an ancestor (descendant resp.) of $\core_2$.
For any bag of $T(G)$ we define its \emph{level} to be the hop distance from the root in $T(G)$.
Also, we define the \emph{graph rooted at $B$} to be the subgraph of $G$ induced by the subtree of $T(G)$ rooted at $B$.
For any core $\core$ we denote by $G_T[\core]$ to be the graph rooted at the center-bag of $\core$.
Let $H(\core)$ denote the set of cores that are ancestors of $\core$ and have rank strictly less than $\core$.
We define the \emph{shadow-domain} of a core $\core$ to be the graph obtained from $G_T[\core]$ by removing the vertices that are contained in at least one core in $H(\core)$.
Note that the shadow domain of a core $R_i$ in phase $2$ is a super-graph of the graph $G_i$.
The \emph{shadow} $\shadow$ of a core $\core$ is defined as the ball of radius $b-a$ centered around $\core$ in the shadow-domain of $\core$.
Note that the shadow of of a core $R_i$ in phase 2 contains $\ball_{G_i}(\core'_i,b-a)$.
The \emph{strict shadow} of $\core$ is defined as its shadow minus itself, ie.~$\shadow \setminus \core$. 

\begin{figure}[h]\label{shadows}
\centering
\begin{subfigure}{.5\textwidth}
  \centering
  \includegraphics[width=1\linewidth]{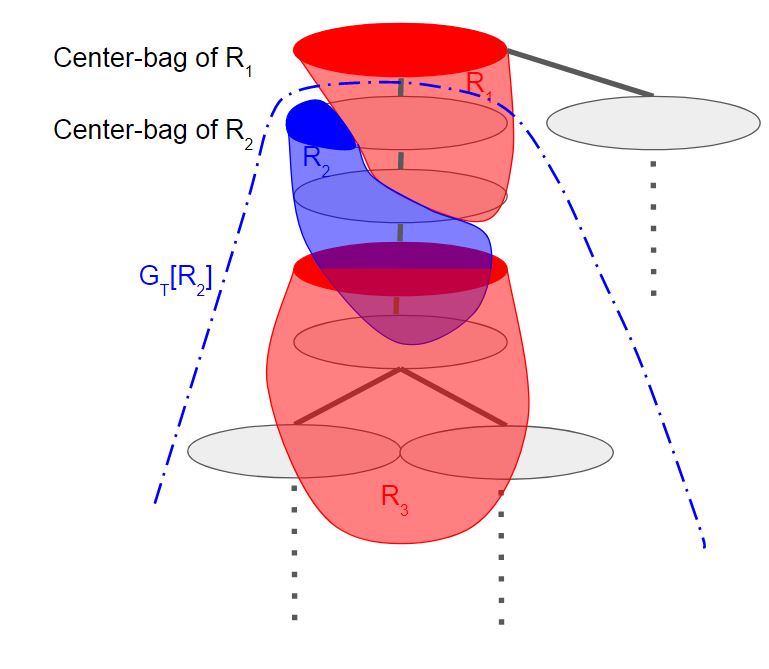}
  \caption{}
  \label{fig:sub1}
\end{subfigure}%
\begin{subfigure}{.5\textwidth}
  \centering
  \includegraphics[width=1\linewidth]{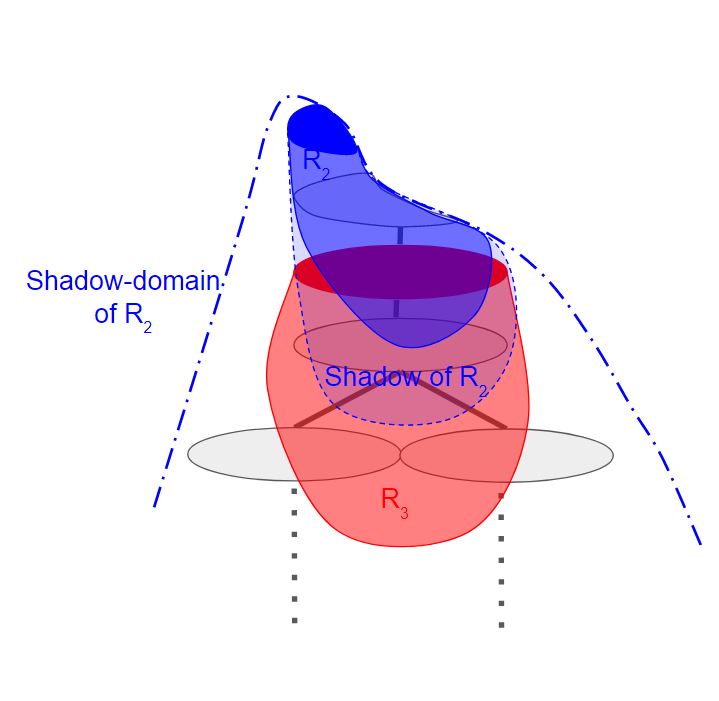}
  \caption{}
  \label{fig:sub2}
\end{subfigure}
\caption{An illustration of the graph $G_T[R]$, shadow-domain, and shadow.}
\label{fig:shadow}
\end{figure}

Next, we will prove some lemmas useful for proving the main lemma of this section.
The following two statements bounds the number of intersecting cores. 


\begin{lemma}
The number of cores intersecting any bag is at most $r^2$.
\label{lem:coresinabag}
\end{lemma}
\begin{proof}
First, let us determine an upper bound for the number of cores of the same rank that can be intersecting a given bag $B \in T(G)$. The cores of the same rank are disjoint by \cref{lemma::coresDisjointRank}. 
Hence there can at most be $\tw$ of them intersecting $B$ because $|B| \leq \tw$. 
Then, since the rank of each core is in $[r]$ by \cref{lem:numberiterations} we have that the total number of cores intersecting any bag is at most $r^2$.
\end{proof}

The following also follows from the Phase 1 algorithm.
\begin{lemma}
\label{lem:noncentervertices}
If $B$ is a bag containing the center $\centre$ of a core that has rank $i$ then all the vertices in 
 $B \setminus \centre$
 have a rank strictly lower than $i$.
\end{lemma}
\begin{proof}
%
In case a bag $B$ contains a center $\centre$ that does not contain all vertices of $B$, ie.~$B \setminus \centre \neq \emptyset$, then as explained above $B \setminus \centre$ have to be covered by the previously constructed cores.
Those cores must have been constructed in previous iterations, by construction.
This implies that the vertices in $\bag\setminus\centre$ have a strictly lower rank than the vertices in $\centre$. 
\end{proof}


The use of \cref{lem:noncentervertices} allows to prove a relation between ranks and shadow-domains.

\begin{lemma}
\label{lem:corewithhigherrank}
Let $\core_1$ and $\core_2$ be two cores such that the rank of $\core_1$ is greater than or equal to that of $\core_2$ and $\core_1$ is an ancestor of $\core_2$.
Let $B$ be the center-bag of $\core_2$ and let $\centre$ be the center of $\core_2$.
If the shadow-domain of $\core_1$ intersects $B\setminus \centre$ then there is a core $\core_3$ that is an ancestor of $\core_2$ and descendant of $\core_1$ and has rank strictly smaller than the rank of $\core_2$ .
\end{lemma}
\begin{proof}
Let $x$ be a vertex in the intersection of the shadow-domain of $\core_1$ and $B \setminus \centre$.
By \cref{lem:noncentervertices} it follows that $x$ has a rank strictly lower than the rank of $\core_2$.
Thus $x$ is contained in a core $\core_3$ that has rank strictly lower than $\core_2$.
Also, $\core_3$ is a descendant of $\core_1$ as otherwise none of the vertices in $\core_3$ and in particular $x$ is not in the shadow-domain of $\core_1$.
Finally, $\core_3$ is an ancestor of $\core_2$ as the center-bag of $\core_2$ contains vertices of $\core_3$.
\end{proof}

The next observation will also help towards proving the main lemma of the section. 

\begin{lemma}
Let $\core_1$ and $\core_2$ be two cores such that the rank of $\core_1$ is strictly greater than that of $\core_2$, and $\core_1$ is an ancestor of $\core_2$.
Furthermore, suppose that there are no cores of rank smaller than $\core_2$ whose center-bag is in the path between the center-bags of $\core_1$ and $\core_2$.
In $G_T[\core_2] $, every vertex in the shadow-domain of $\core_1$ is in the shadow-domain of $\core_2$.
\label{lem:shadowGraphInclusion}
\end{lemma}
\begin{proof}
Let $x$ be a vertex in $G_T[\core_2]$ that is not in the shadow-domain of $\core_2$.
It suffices to prove that $x$ is not in the shadow-domain of $\core_1$.
Since $x$ is not in the shadow-domain of $\core_2$ there should be an ancestor core $\core_3$ of $\core_2$ having rank smaller than $\core_2$ containing $x$.
By the precondition of the lemma, the core $\core_3$ has to be an ancestor of also $\core_1$.
This means that $x$ cannot be in the shadow-domain of $\core_1$.
\end{proof}

We use all the above properties to prove the following lemma, from which the main lemma of the section follows rather directly.

\begin{lemma}
\label{lem:vertexshadowbound}
If $b=2a$ in our algorithm, then for any vertex $u$, the number of cores whose strict shadow contains $u$ is at most $r^3$.
\end{lemma}
\begin{proof}
Consider a vertex $u$.
Let $\C$ be the set of cores whose strict shadow contains $u$.
So, our goal is to prove that $|\C|\le r^3$.

Let $P$ be the path in $T(G)$ from the root bag to the bag containing $u$. 
The center-bags of the cores in $\C$ are all on $P$.
This comes from the fact that 
 the shadow-domain of a core is a subgraph of the graph rooted at its center-bag.

Let $r_1$ be the lowest rank among all cores in $\C$.
We will show that there is only one core in $\C$ having rank $r_1$ in \cref{lem:oner_j}.
Let this unique core in $\C$ with rank $r_1$ be $P_1$.
Let $\C_1$ be the set of all cores in $\C$ (including $P_1$) that are ancestors of $P_1$. 
We will show that each core in $\C_1$ intersects the center of $P_1$ in \cref{lem:intersectQ_j}. 
Since the center is contained in a bag, by \cref{lem:coresinabag} this implies that $|\C_1|\le r^2$.

Now, let $\bar{\C}_1=\C\setminus \C_1$.
If $\barC_1$ is non-empty, define $r_2$ be the lowest rank among cores in $\bar{\C_1}$.
Note that $r_2>r_1$.
We show that there is only one core in $\barC_1$ having rank $r_2$ in \cref{lem:oner_j}.
Let this unique core in $\barC_1$ with rank $r_2$ be $P_2$.
Let $\C_2$ be the set of all ancestor cores $P_2$ in $\barC_1$ including $P_2$.
We show that each core in $\C_2$ intersects the center of $P_2$ in \cref{lem:intersectQ_j}.
Since the center is contained in a bag, by \cref{lem:coresinabag} this implies that $|\C_2|\le r^2$.

We repeat the procedure and define the sequences $\C_1,\C_2,\C_3,\dots ,\C_{\ell}$  and $r_1<r_2<\dots <r_{\ell}$ until $\barC_{\ell}=\barC_{\ell-1}\setminus \C_{\ell}$ is empty, where $\barC_0 = \C$.
Here, $r_j$ is defined as the lowest rank among cores in $\barC_{j-1}$.
For each $j\in [\ell]$, we show that there is a unique core $P_j$ in $\barC_{j-1}$ with rank $r_j$ in \cref{lem:oner_j}. 
Let $\C_j$ be the set of all cores in $\barC_j$ that are ancestors of $P_j$ (including $P_j$).
Also, we show that all cores in $\C_j$ intersects the center of $P_j$ in \cref{lem:intersectQ_j}, 
implying that $|\C_j|\le r^2$ for each $j\in [\ell]$.
Since $r_1<r_2<\dots <r_{\ell}\le r$,
this implies that $|\C|\le r^3$.

Thus, the proof of the lemma concludes by proving the following claims.
\begin{claim}
There is only one core in $\barC_{j-1}$ having rank $r_j$ for each $j\in [\ell]$.
\label{lem:oner_j}
\end{claim}
\begin{proof}
Suppose this is not true.
Then there are two cores $\core_1$ and $\core_2$ in $\barC_{j-1}$ having rank $r_j$.
Assume without loss of generality that $\core_1$ is an ancestor of $\core_2$.
Note that $\core_1$ and $\core_2$ are disjoint as they have the same rank.
Also, since $\core_1,\core_2\in \C$, the center-bags of $\core_1$ and $\core_2$ lie on the path $P$.


Since $u$ is in the shadow of $\core_1$ there is a path $Z$ of length at most $b$  that goes from the center of $\core_1$ to $u$ in the shadow-domain of $\core_1$.
This path $Z$ has to intersect the center-bag of $\core_2$ to get to $u$.

Suppose this intersection occurs at a vertex in the center of $\core_2$.
Then the path $Z$ goes from the center of $\core_1$ to outside $\core_1$ and then into the center of $\core_2$ and then to outside of $\core_2$. It has to go outside of $\core_2$ as $u$ is in the strict shadow of $\core_2$.
Also, it has to go outside of $\core_1$ before entering $\core_2$ as $\core_1$ and $\core_2$ are disjoint (however, it is possible that there is an edge from $\core_1$ to $\core_2$).
This means $Z$ has a length of more than $a+a=2a=b$, a contradiction.

Now, suppose this intersection occurs at a vertex in the center-bag of $\core_2$ that is not in the center of $\core_2$.
Then by \cref{lem:corewithhigherrank} it follows that there is at least one core that is ancestor of $\core_2$ and descendant of $\core_1$, and having rank lower than $r_j$.
Let $\core_3$ be the one among such cores whose center-bag has the smallest level. 

We claim that the path $Z$ intersects the center of $\core_3$.
Suppose otherwise.
However, the path has to intersect the center-bag of $\core_3$ to get to $u$.
Then, by \cref{lem:corewithhigherrank} it follows that there is a core that is ancestor of $\core_3$ and descendant of $\core_1$, and having rank lower than $r_j$, a contradiction to the selection of $\core_3$.

Note that $u$ is not in $\core_3$ as $\core_3$ is disjoint from the shadow-domain of $\core_2$ by definition of shadow-domain.
However, $u$ is contained in the shadow of $\core_3$ as the part of the path $Z$ from the center of $\core_3$ to $u$ is contained in the shadow-domain of $\core_3$ (by \cref{lem:shadowGraphInclusion}) and has length at most $b$.
Thus, $u$ is contained in the strict shadow of $\core_3$.
Since $\core_3$ is a descendant of $\core_1$, we have that $\core_3\notin \C_1\cup\C_2\cup\dots\cup\C_{j-1}$ and hence $\core_3\in \barC_{j-1}$.
Thus, there is a core in $\barC_{j-1}$ that has rank lower than $r_j$, a contradiction to the choice of $r_j$.
\end{proof}

\begin{claim}
Each core in $\C_j$ intersects the center of $P_j$ for each $j\in [\ell]$. 
\label{lem:intersectQ_j}
\end{claim}
\begin{proof}
Suppose this is not true.
Then there is a core $\core \in \C_j$ that is disjoint from the center of $P_j$.
Note that $P_j$ is a descendant of $\core$ by definition of $\C_j$ and the rank of $P_j$ is less than the rank of $\core$ by the definition of $P_j$.
Since $u$ is in the shadow of $\core$ there is a path $Z$ of length at most $b$ that goes from the center of $\core$ to $u$ in the shadow-domain of $\core$.
This path $Z$ has to intersect the center-bag of $P_j$ to get to $u$.

Suppose this intersection occurs at a vertex in the center of $P_j$.
Then the path $Z$ goes from the center of $\core$ outside $\core$ and then into the center of $P_j$ and then to outside of $P_j$. It has to go outside of $P_j$ as $u$ is in the strict shadow of $P_j$.
Also, it has to go outside of $\core$ before entering $P_j$ as $\core$ and the center of $P_j$ are disjoint. 
This means $Z$ has a length of more than $a+a=2a=b$, a contradiction.

Now, suppose this intersection occurs at a vertex in the center-bag of $P_j$ that is not in the center of $P_j$.
Then, by \cref{lem:corewithhigherrank}, there exist at least one core that is an ancestor of $P_j$ and a descendant of $\core$ and having rank lower than $r_j$.
Let $\core'$ be the one among such cores whose center-bag has the smallest level.
The path $Z$ has to intersect the center of $\core'$ as otherwise
there is a core that is an ancestor of $\core'$ and a descendant of $\core$ and having rank lower than $r_j$, contradicting the selection of $\core'$.

Note that $u$ is not in $\core'$ as $\core'$ is disjoint from the shadow-domain of $\core$ by definition of shadow-domain.
However, $u$ is contained in the shadow of $\core'$ as the the part of the path $Z$ from center of $\core'$ to $u$ is contained in the shadow-domain of $\core'$ (by \cref{lem:shadowGraphInclusion}) and has length at most $b$.
Thus, $u$ is contained in the strict shadow of $\core'$.
This implies that $\core'\in \C_j$.
Thus we have a core in $\C_j$ having rank strictly smaller than $r_j$, a contradiction.
\end{proof}\qedhere

\begin{proof}[Proof of \cref{lem:edgeVolbound}]
Each vertex $u$ appears in the strict shadow of at most $r^3$ cores by \cref{lem:vertexshadowbound}. Also, any vertex can only intersect one core of a fixed rank, and hence the total number of cores intersecting a vertex $u$ is at most $r$. Thus, each vertex $u$ appears in the shadow of at most $r^3+r$ cores. This implies that for any edge $(u,v)$, the number of cores such that at least one of $u$ or $v$ appear in the shadow is at most $2r^3+2r$.
Observe that by construction, the shadow of a core $R_i$ contains the ball $\ball_{G_i}(\core'_i,b-a)$.
This implies the lemma.
\end{proof}
\end{proof}
\section{Bounding the Total Weight of Cut Edges}\label{section::boundWeight}

\begin{lemma} \label{lemma:phase_2_bound_weight}
If $b=2a>0$ in our algorithm, then the total capacity of the cut edges $X_2$ picked in the second phase is at most $\dfrac{16}{a}\cdot \ln (\tw+1) \cdot F^*$, where $F^*= \sum_{e \in E}c(e) \cdot l(e)$. 


\end{lemma}
\begin{proof}
Let $R_i,S_i,G_i,R'_i,t_i$ be as in Phase 2 description.
Note that the set of cut-edges $X_2$ in Phase $2$ is equal to $\bigcup\delta_{G_i}(S_i)$.
Thus the cost of $X_2$ can be bounded by $\sum_{i \in [p]} C_{G_i}(\core'_i,\ra_i)$. 
Recall that we selected $t_i$ in Phase 2 algorithm such that
    \begin{align*}
    C_{G_i}(\core'_i,\ra_i)  &\leq  \dfrac{1}{b-a} \cdot \ln  \left( \dfrac{\vol_{G_i}(R'_i,b-a)}{\vol(R'_i,0)} \right) \cdot \vol_{G_i}(\core_i',\ra_i)\\
    \end{align*}
Also, recall that for a core $\core_i$, the initial volume $\vol(R'_i, 0)$ was chosen to be $\vol'_{G_i}(R'_i,b-a)/h$ where
$h=2\tw^3 + 2\tw$. 
By \cref{lem:edgeVolbound}, any edge $(u,v)$ contributes at most $h$ times to the sum $\sum_{i\in [p]}\vol'_{G_i}\left(R'_i,b-a\right)$ and we obtain the following: 

\begin{claim} \label{claim: coreset_bound_volume}
    $\sum_{i\in [p]}\dfrac{\vol'_{G_i}\left(R'_i,b-a\right)}{h} \leq \sum_{e \in E} c(e) \cdot l(e) = F^*$
\end{claim}


Thus,
    \begin{align*}
    C_{G_i}(\core'_i,\ra_i)  &\leq  \dfrac{1}{b-a} \cdot \ln  \left( \dfrac{\vol_{G_i}(R'_i,b-a)}{\vol(R'_i,0)} \right) \cdot \vol_{G_i}(\core_i',\ra_i)\\
    & =  \dfrac{1}{b-a} \cdot \ln  \left( \dfrac{\vol(R'_i,0)+\vol'_{G_i}(R'_i,b-a)}{\vol(R'_i,0)} \right) \cdot \vol_{G_i}(\core_i',\ra_i)\\
    & =  \dfrac{1}{b-a} \cdot \ln  \left( 1+\dfrac{\vol'_{G_i}(R'_i,b-a)}{\vol'_{G_i}(R'_i,b-a)/h} \right) \cdot \vol_{G_i}(\core_i',\ra_i)\\
    & =  \dfrac{1}{b-a} \cdot \ln  \left( 1+h \right) \cdot \vol_{G_i}(\core_i',\ra_i)\\
    \end{align*}
Summing now up over all coresets yields:
\begin{align*}
    \sum_{i \in [p]} C_{G_i}(\core'_i,\ra_i)
    &\leq \sum_{i \in [p]}\dfrac{1}{b-a} \cdot \ln  \left( 1+h \right) \cdot \vol_{G_i}(\core_i',\ra_i)\\
    &= \dfrac{1}{a} \cdot \ln \left(2r^3 + 2r + 1\right) \cdot \sum_{i \in [p]}\vol_{G_i}(\core_i',\ra_i) \\
    &\leq \dfrac{8 \cdot \ln (r+1)}{a} \cdot \sum_{i \in [p]}\left(\dfrac{\vol'_{G_i}(\core_i',b-a)}{h} + \vol'_{G_i}(\core_i',\ra_i)\right)\\
    &\leq \dfrac{8 \cdot \ln (r+1)}{a} \cdot ( F^* +F^*)=\dfrac{16 \cdot \ln (r+1)}{a} \cdot F^*
\end{align*}


The last inequality follows from Claim \ref{claim: coreset_bound_volume} and the fact that any edge contributes to at most one term in $\vol'_{G_i}\left(\core_i',\ra_i)\right)$.
\end{proof}
\textbf{Proof of Theorem \ref{Theorem: main}} : 
We showed in \cref{lem:smalldiameterdecomp} that $X_2\cup X_3$ gives a small diameter decomposition.
It only remains to bound the cost of $X_2\cup X_3$.
Lemma \ref{lemma:phase_2_bound_weight} gives a bound on the total capacity of edges in $X_2$. 
It follows directly from Lemma \ref{lemma:sdd_width_0} the capacity of $X_3$ used in Phase 3 is at most $8 \cdot \ln (r+1) \cdot F^*$. Hence the total cost of the small diameter decomposition is at most $8 \cdot \ln (r+1) \cdot F^*+ 128 \cdot \ln (r+1) \cdot F^*=136 \cdot \ln (r+1) \cdot F^*$. This completes the proof of Theorem \ref{Theorem: main}.

\section{Concluding Remarks}
In this paper, we give an algorithm for finding low diameter decomposition of small cost for bounded treewidth graphs, more specifically with a cost logarithmic in treewidth.
Our result also imples a multiflow-multicut gap and an approximation ratio for multicut that is logarithmic in treewidth.
These results are tight asymptotically.
We believe that our techniques could provide useful insights for proving such a (tight) result for the class of $K_r$-minor-free graphs. We also believe that it should be possible to extend our techniques for constructing low-diameter padded decomposition for bounded treewidth graphs (see \cite{abraham2014cops} for the definition of padded decomposition). 

    \bibstyle{alpha}
    \bibliography{main}

\end{document}